\documentclass[12pt]{article}
\usepackage{amsmath,amssymb,amsthm}
\usepackage{graphicx}
\usepackage{epstopdf}
\usepackage{hyperref}
\usepackage{natbib}

\renewcommand{\)}{\right)}
\renewcommand{\(}{\left(}

\newtheorem{algorithm}{Algorithm}{\bf}{\it}
\newtheorem{theorem}{Theorem}
\newtheorem{lemma}{Lemma}
\theoremstyle{definition}
\newtheorem{definition}{Definition}
\theoremstyle{remark}
\newtheorem{example}{Example}
\title{Secretary problem: graphs, matroids and greedoids}
\author{Wojciech Kordecki \\
The Witelon State University of Applied Sciences in Legnica\\ 
e-mail: wojciech.kordecki@pwsz-legnica.eu
} 

\date{\today
} 

\usepackage{mathptmx} 

\begin{document}

\maketitle

\begin{abstract}
In the paper the generalisation of the well known ``secretary problem'' is considered. 
The aim of the paper is to give a generalised model in such a way that the chosen set of the possible best $k$ elements have to be independent of all rejected elements. This condition is formulated using the theory of greedoids and in their special cases -- matroids and antimatroids. Examples of some special cases of greedoids (uniform, graphical matroids and binary trees) are considered.

\end{abstract}

\section{Introduction}
\label{s:intro}

The \textit{secretary problem} also known as the \textit{marriage problem} relies on a choice of the best candidate in such a way that only the relation to the previously interviewed candidates is known and the rejected candidates are definitively lost. The number of candidates is also known before the interview starts. Then after the interview we have to decide whether to accept the candidate or not? 
Our goal is to choose the best candidate, i.e. we have to decide when the process of recruitment should be stopped.
In a more general situation we want to choose not only one, the best candidate, but we want to choose the best $k$ members who form a team. 

In the simplest case we do not have any limitation given to recruitment process or the relationships inside the team. 
In this paper we focus our attention on the limitation of a recruitment process. 
namely we can choose only such candidates who are not dependent on the candidates rejected in the current interview.

Such an idea of a recruitment process requires a precise explanation of the meaning of the sentence ``independent of previously rejected candidates''. 
As the next step we have to determine the stopping rule to obtain the optimal stopping time. 
The main aim of this paper is to formulate a sufficiently general but practicably useful structure of dependence. 

Finding the optimal solution of the problem described above in the general case seems impossible in general cases. 
Therefore we study optimal algorithms for finding the best solution only in the some particular, but apparently useful cases.

The paper is organised as follows. In Section~\ref{s:s-problem} the classical secretary problem is introduced. 
Next, the  variant of this problem with the the necessary independence between rejected candidates and accepted ones is presented. 
Section~\ref{s:m-g} introduces the most known independence structures: matroids and their generalisation -- greedoids. 
At the end of that section, the problem in the general greedoid case, is introduced. 
In Section~\ref{s:special} some particular, selected models are introduced. 
In the simpler models, the solutions are given. 
In the more complicated models only some connections between known results (for example from random graph theory) and problems of optimal stopping in such models are discussed. 

\section{Secretary problem}
\label{s:s-problem}

\subsection{Classical secretary problem}
\label{ss:s-problem}

In the classical secretary problem there are $n$ linearly ordered elements $\{1,2,\dots,n\}$.
They are being observed at a random order $\(e_1, e_2,\dots,e_n\)$. 
At the moment $t=i$ the observer knows only the relative ranks of the elements $e_t$ examined so far.
Once rejected, an element cannot be recalled.

The aim of the observer is to choose the currently examined object in such a way that the probability $\Pr\(e_t = n\)$ will be maximal.

This problem is well known and solved. 
Dynkin in \citeyear{Dynkin:opt} shows that for large $n$, it is approximately optimal to wait until a fraction $1/e$ of the elements appears and then to select the next relatively best one. 
The probability of success is also $1/e$.
More strictly, we can present this result as follows.
Let $w\(e\)$ denote the rank of $e$ and $w\(A\)=\max_{e\in A}w\(e\)$.
\begin{theorem}\label{thm:Dynkin}
Let us assume that an algorithm of choices has the following form.
\begin{enumerate}
\item
Reject all elements $e_t$ for subsequent $t\leq v$  for some $v$.
\item 
If $t>v$ then we accept $e_t$ if $w\(e_t\)\geq w\({A_{t-1}}\)$ or reject it in the opposite case. The rejection is irrevocable.
\item 
The process is stopped if the element is accepted or $t=n$.
\end{enumerate}
If $v\sim n/e$ with $n\to\infty$ then the $\Pr\(e_t = n\)$ is maximal and is equal to $1/e$.
\end{theorem}
The easy proof of Theorem~\ref{thm:Dynkin} is a good pattern for considerations which will be used in more general models given in the next parts of this article. 
Therefore, this proof is presented in a more detailed way than it is required in this particular case.
\begin{proof} (see \cite{Ferguson:Who})
Assume that the first $v-1$ elements are rejected and element $m$ has the highest rank among these $v-1$ elements. 
Next, select the first subsequent element that is better than element $m$.
For an arbitrary $v$, the probability that the element with the highest rank is selected is
\[
\begin{split}
P\(r\)
=&\sum_{i=1}^n\Pr\(\text{element $i$ is selected}\cap\text{element $i$ has the highest rank}\) \\
=&\sum_{i=1}^n\Pr\(\text{element $i$ is selected}|\text{element $i$ has the highest rank}\) \\
&\times\Pr\(\text{element $i$ has the highest rank}\) \\
=&\frac{1}{n}\sum_{i=v}^n
\Pr\(\left.\text{\begin{minipage}{4.2cm}
the element with the highest rank of the first $i-1$ elements is in the first $v-1$ elements
\end{minipage}
}\right|
\text{\begin{minipage}{3.5cm}
element $i$ has the highest rank
\end{minipage}
}
\) \\
=&\frac{v-1}{n}\sum_{i-1}^n\frac{1}{i-1}
\end{split}
\]
Therefore the best choice is with probability:
\[
\frac{1}{n}\sum_{i=v}^n\frac{v}{i}
\approx\int\limits_v^n\frac{dx}{x}
=\frac{v}{n}\ln\frac{n}{v}.
\]
The the maximum is achieved for $v\approx n/e$. \qed
\end{proof}
See \cite{Ferguson:Who} for a brief historical review of this classical secretary problem.

An important generalisation of this problem is known as the multiple choice secretary problem (see ~\cite{Hajiaghayi_et_al:Adptive}, \cite{Kle05},  \cite{GirdharDudek:How}). 
The objective of this problem is to select a group of at most $k$ secretaries from a pool of $n$ applicants having a combined value as large as possible.

\subsection{Secretary problem and independence}
\label{ss:ms-problem}

Our generalisation leaves a linear order but assumes an additional combinatorial structure in the set of elements $e_i$.
Using the language of the optimal choice of the candidate to a position (secretary problem), our problem can be described as follows. 

\emph{The subsequent candidates arrive. 
We can reject the candidate and then we consider a new candidate. 
The rejected candidate is irretrievably lost. 
Every new candidate is compared to the previously rejected candidates. 
If a new candidate is dependent on the previously rejected ones, such a candidate is also rejected. 
If the candidate is not dependent, then as a result of the comparison we can reject or accept him/her.
}

The main aim of the article is the research of stopping criteria if the random variables are indexed by elements of a finite structure and the permissible choice is limited by such a structure. 
Assume tentatively that an element $e$ is independent on the set $A$ if it does not belong to the closure of $A$. The name `closure' needs defining which will be done in the next sections. 
Our basic assumptions are:
\begin{itemize}
\item 
in the structure, a closure operator and a family of closed sets are specified,
\item 
if a new element belongs to the closure of previously rejected elements, then it also has to be  rejected,
\item 
if it does not belong to the closure, the new element can be accepted.
\end{itemize}

Let us consider a simple, but illustrative example. 
The structure in this example is known as ``linear structure'' which is a special case of ``strictly hierarchical structure'' (see~\cite{Klimesch:Structure}, p.~46).
At first we have to formulate the following simple combinatorial result.
\begin{lemma}
Denote $\left[n\right]=\left\{1,2,\dots,n\right\}$ and let $j\in\left[n\right]$ be fixed. 
The number $P\(j,n\)$ of permutation $\pi:\left[n\right]\to\left[n\right]$ such that 
\begin{equation}\label{eq:cond_j}
k>j\implies\pi\(k\)>\pi\(j\)
\end{equation}
is equal 
\begin{equation}\label{eq:permut_j}
P\(j,n\)=\frac{n!}{n-j+1}
\end{equation}
\end{lemma}
\begin{proof}
For $i=\pi\(j\)$ the number of permutations fulfilling \eqref{eq:cond_j} is equal
\[
\begin{split}
\binom{n-i}{n-j}\(n-j\)!\(j-1\)!
&=\frac{\(n-i\)!}{\(n-j\)!\(j-i\)!}\(n-j\)!\(j-1\)! \\[0.3ex]
&=\frac{\(n-i\)!\(j-1\)!}{\(j-i\)!}\,.
\end{split}
\]
For all $i\leq j$ we obtain
\[
\sum_{i=1}^j\frac{\(n-i\)!\(j-1\)!}{\(j-i\)!}=\frac{n!}{n-j+1}\,.
\]
which completes the proof.\qed
\end{proof}

\begin{example}\label{xmpl:simplest}
Let $S=\left[n\right]=\{1,2,\dots,n\}$.
We will make the following assumptions: every secretary $s\in S$ has two features -- qualification (weight) $w\(s\)$ and position in the hierarchical  organisation (i.e.~rank) $r\(s\)$. Let all weights $w\(s\)$ and ranks $r\(s\)$ will be different. 
Let 
\[
r\(S\)=\max\{r\(s\):s\in S\}
\]
and 
\[
w\(S\)=\max\{w\(s\):s\in S\}.
\]
If $S$ is the set of candidates rejected so far, then if $r\(t\)<r\(S\)$ for a new $t\notin S$, $t$ has to be rejected even if $w\(t\)>w\(S\)$. 
In other words, having rejected the boss we must not employ the subordinate. More formally, the element $t$ is independent of the $S$ if $r\(t\)>r\(S\)$.
Note however, that at the moment $t$ we do not know the values $r\(t\)$ and $w\(t\)$ but we can only verify if the inequalities $r\(t\)>r\(S\)$ and $w\(t\)>w\(S\)$ are fulfilled.

In this example we consider two completely different cases.
The first ideal case: 
\[
r\(s_1\)>r\(s_2\)\iff w\(s_1\)>w\(s_2\). 
\]
Then we can assume that $r\(i\)=w\(i\)=i$.
This case coincides with the classical secretary problem.

The second is the most haphazard case: weight and rank are independent random variables%
\footnote{\label{footnote:sim}Any similarity to actual events is purely coincidental.}. 
Then we can assume that $r\(i\)=i$ but $w\(i\)=\pi\(i\)$ where $\pi$ is random permutation of $\left[n\right]$.

In this case let us try to pick the best candidates in the same way as in the classic problem. 
First we examine and reject a fraction $\alpha$ of $n$ candidates (say $S_{\alpha}$) and at the next steps $k>n/\alpha$ we pick the first candidate $s_k$ with the rank and weight higher of the candidates rejected so far, i.e. $r\(s_k\)>r\(S_{k-1}\)$ and $w\(s_k\)>w\(S_{k-1}\)$. 

Let us denote the most valuable candidate by $z_1$ and the second most valuable by $z_2$. 
If $z_2\in S_{\alpha}$, $z_1\notin S_{\alpha}$ and moreover $r\(z_1\)>r\(S_{\alpha}\)$, the selected candidate is the best.
Therefore the probability that the randomly chosen permutation fulfils \eqref{eq:cond_j} for given $n$ and $j>k_0$ for some fixed $k_0$ is equal
\begin{equation}\label{eq:prob_j}
\frac{1}{n}\sum_{j=r+1}^n\frac{1}{n-j+1}=\frac{1}{n}\sum_{j=1}^{n-r}\frac{1}{j}
\sim \frac{\ln\(n-k_0\)+\gamma}{n},
\end{equation}
where $\gamma=0.5772156649\dots$ is an Euler constant.

Continuing this example for the second, haphazard case, let us assume that $\alpha=1/2$. 
In such the case, let $z_1\in S_{1/2}$ and $z_2\notin S_{1/2}$. 
The candidate $z_1$ is elective if $r\(j\)<r\(z_1\)$ for all $j<z_1$. 
From
\begin{equation}\label{eq:alpha2-1}
\Pr\(z_1\in S_2\text{ and }z_2\notin S_2\)=\frac{1}{4}
\end{equation}
and from Equation~\eqref{eq:prob_j} we obtain 
\begin{equation}\label{eq:alpha2-2}
\Pr\(r\(j\)<r\(z_1\)\)=\frac{1}{4n}\sum_{j=1}^{n/2}\frac{1}{j}>\frac{\ln\(n/2\)}{4n}=R_{1/2}\,.
\end{equation}

It seems that a better way is to take as $\alpha$ the value other than $1/2$. 
Then  we have 
\begin{equation}\label{eq:alpha-1}
\Pr\(z_1\in S_2\text{ and }z_2\notin S_2\)=\alpha\(1-\alpha\)
\end{equation}
and instead of \eqref{eq:alpha2-2} we obtain
\begin{equation}\label{eq:alpha-2}
\Pr\(r\(j\)<r\(z_1\)\)=\alpha\(1-\alpha\)\frac{1}{n}\sum_{j=1}^{\alpha n}\frac{1}{j}
>\alpha\(1-\alpha\)\ln\(\alpha n\)/n=R_{\alpha}\,.
\end{equation}
Nevertheless the difference between the right side of \eqref{eq:alpha2-2} and the maximal value of the right side of \eqref{eq:alpha-2} is very small, less than $1\%$ -- see Table~\ref{tab:rhs-alpha}.

\begin{table}[!hbt]
\centering 
\caption{\label{tab:rhs-alpha}Right-hand side of \eqref{eq:alpha2-2} and maximum of right-hand side of \eqref{eq:alpha-2}}

\smallskip
\begin{tabular}{r|c|c|c|c}
$n$ & $\alpha$ & $R_{1/2}$ & $\max R_{\alpha}$ &
$\frac{\max R_{\alpha}-R_2}{\max R_{\alpha}}$ \\\hline 
10 &  0.6084 &    0.04024 & 0.04024 & 0.0647 \\
20 &  0.5844 &    0.02878 & 0.02986 & 0.0356  \\
50 &  0.5653 &   0.01609 & 0.01642 & 0.0200  \\
100 & 0.5555 &  0.00978 & 0.00992 & 0.0141 
\end{tabular}
\end{table}
Note that in the ludicrous situation\footnote{See footnote~\ref{footnote:sim}}
\[
r\(e_1\)<r\(e_2\)\iff w\(e_1\)>w\(e_2\)
\]
for any pair $e_1,e_2$, the optimal strategy is to choose the first candidate. 
Every next candidate will be either worse or dependent.
This situation leads of course, with high probability $\(n-1\)/n$, to the lack of choice, so it can be neglected.

As the third case in this example we can consider such a situation that the correlation between ranks and weights is positive (usually essentially greater than zero), but smaller than one. 
Such a case needs more precise assumptions and probabilistic considerations hence it will be omitted in this paper.
\end{example}

\section{Matroids and greedoids}
\label{s:m-g}

As it was mentioned previously we need a precise definition of the words `closure' of $A$ and `independent' element $e$ from the set $A$.  
The useful tool to give such the definitions are structures known as matroids and more generally -- greedoids. 
In the next two sections we provide the necessary definitions and results from the matroid and greedoid theory.

\subsection{Matroids}
\label{ss:ma}

Let $E$ be a finite set. A family $\mathcal{I}$ of subsets of $E$ is the family of independent sets if the following conditions hold:
\begin{description}
\item[$(i_1)$]
$\emptyset\in\mathcal{I}$,
\item[$(i_2)$]
if $I_1\subseteq I_2\in\mathcal{I}$, then $I_1\in\mathcal{I}$,
\item[$(i_3)$]
if $I_1,i_2\in\mathcal{I}$, $|I_1|<|I_2|$, then there exists $e\in I_2\setminus I_1$,
such that $I_1\cup\{e\}\in\mathcal{I}$.
\end{description}
A pair $\(E,\mathcal{I}\)$ is a matroid (see for example \cite{Oxley:ma2nd}, \cite{Welsh:ma}, \cite{Wilson-eng5}). 

A basis is every maximal independent set. All bases have the same number of elements. A rank $\rho\(A\)$ of any set $A\subseteq E$ is the number of elements of maximal independent set $I\subseteq A$. A closure $\sigma\(A\)$ of a set $A$ is the maximal set with the same rank as $A$. The set $A$ is closed if $\sigma\(A\)=A$.
The operator $\sigma$ for matroids fulfils the following properties:
\begin{description}
\item[$(s_1)$]
$A\subseteq\sigma\(A\)$,
\item[$(s_2)$]
if $A\subseteq B$ then $\sigma\(A\)\subseteq\sigma\(B\)$,
\item[$(s_3)$]
$\sigma\(\sigma\(A\)\)=\sigma\(A\)$,
\end{description}

Using the definition of matroid, we can interpret ``an independence'' of element $e$ of the set $A$ in such a way that $e\notin\sigma\(A\)$. 
Comparing this interpretation with the example in Section~\ref{ss:ms-problem}, we can remark that such a meaning of independence is not fortunate because the closure $\sigma\(A\)$ has the \textit{exchange property}: 
\begin{description}
\item[$(ex)$]
if $f\notin\sigma\(A\)$, $f\in\sigma\(A\cup\left\{e\right\}\)$ then $e\in\sigma\(A\cup\left\{f\right\}\)$.
\end{description}
A structure $\(E,\mathcal{I}\)$ is a matroid if and only if $\sigma$ fulfils the conditions $(s_1)$ and $(s_2)$ and the condition $(ex)$. Note that $(s_3)$ follows from $(s_1)$ and $(s_2)$ and the condition $(ex)$ but $(s_1)$ -- $(s_3)$ does not give $(ex)$. Therefore the set of conditions  $(s_1)$ -- $(s_3)$ is not a characterisation of a matroid.

\subsection{Greedoids}
\label{ss:gree}

\subsubsection{Basic definitions and properties}
\label{ss:basic}

The hierarchical structure of dependence in Example~\ref{xmpl:simplest} does not fulfil the condition $(ex)$. Therefore we have to use more a general structure than matroids.

A greedoid (a greedy structure) is the family $\mathcal{F}$ of subsets of the set $E$ which fullfils the following conditions (see for example \cite{Korte_et_al:Greedoids}, \cite{KorteVygen:CombOpti5}):
\begin{description}
\item[$(f_1)$]
$\emptyset\in\mathcal{F}$,
\item[$(f_2)$]
if $F_1,F_2\in\mathcal{F}$, $|F_1|<|F_2|$, then there exists $e\in F_2\setminus F_1$, such that $F_1\cup\{e\}\in\mathcal{F}$.
\end{description}
Note that the conditions for greedoids are the conditions for matroids with the exception of $(i_2)$. The family $\mathcal{F}$ is called \textit{feasible}.
The family $\mathcal{F}$ is called \textit{accessible} if the following condition holds:
\begin{description}
\item[$(a_1)$] 
if $F\in\mathcal{F}\setminus\{\emptyset\}$ then there exist $s\in F$ such that $F\setminus\{e\}\in\mathcal{F}$. 
\end{description}
The pair $\(E,\mathcal{F}\)$ where $\mathcal{F}$ is accessible is called an accessible system. Every greedoid is an accessible system. Matroids are also greedoids with independent sets as feasible sets.
Clearly, the property $(a_1)$ is weaker than the property $(i_2)$ -- does not every subset of an independent set is independent, but at least one subset of a feasible set is also feasible.

A basis is every maximal feasible set. 
All bases have the same number of elements. A rank $\rho\(A\)$ of any set $A\subseteq E$ is the number of elements of maximal feasible set $F\subseteq A$. 
A closure $\tau\(A\)$ of a set $A$ is the maximal set with the same rank as $A$, i.e. (see~\cite{KorteLovas:structural} or~\cite{Korte_et_al:Greedoids})
\begin{equation}\label{eq:greedy-closure}
\tau\(A\)=\left\{x\in E:\rho\(a\cup\left\{x\right\}\)=\rho\(A\)\right\}.
\end{equation}
The closure $\tau\(A\)$ defined by~\eqref{eq:greedy-closure} fulfils the conditions $(s_1)$ and $(s_3)$ but not necessarily the condition $(s_2)$, i.e. closure operator is not necessarily monotone (see~\cite{Korte_et_al:Greedoids}, Example on p.~~69, fig.~6). However one can define the monotone closure operator $\sigma\(A\)$:
\begin{equation}\label{eq:greedy-monotone-closure}
\sigma\(A\)=\bigcap\left\{X:A\subseteq X, \tau\(X\)=X\right\}.
\end{equation}
It is easy to see that the monotone closure $\sigma\(A\)$ satisfies all conditions $(s_1)$ -- $(s_3)$, but greedoids is not uniquely determined by its monotone closure operator (see~\cite{Korte_et_al:Greedoids}, p.~63).

If a greedoid fulfils the \textit{antiexchange} property
\begin{description}
\item[$(aex)$]
if $f\notin\sigma\(A\)$, $f\in\sigma\(A\cup\left\{e\right\}\)$, $f\neq e$ then $e\notin\sigma\(A\cup\left\{f\right\}\)$
\end{description}
then we call such a greedoid an antimatroid. 
\begin{theorem}[\cite{KorteVygen:CombOpti5}, Th.~14.4]
\label{thm:KV1}
If $\(E,\mathcal{F}\)$ is an antimatroid then 
\begin{equation}\label{eq:convex}
\sigma\(A\)=\bigcap\left\{X\subseteq V:A\subseteq X,V\setminus X\in\mathcal{F}\right\} 
\end{equation}
is a closure operator, i.e. it satisfies conditions $(s_1)$-- $(s_3)$.
\end{theorem}

The structure of the Example~\ref{xmpl:simplest} is an antimatroid if we take as closed sets all the sets of the form $\left[k\right]$, where $\left[k\right]=\left\{1,\dots,k\right\}$ for $1\leq k\leq n$ and $\left[0\right]=\emptyset$. The feasible sets have the form $\left[n\right]\setminus\left[k\right]$ for $0=\leq k\leq n$.

\begin{lemma}\label{lem:anti}
Let $\(E,\sigma\)$ be an antimatroid. Suppose that the sequence $e_1,\dots,e_n$ is such that
\begin{equation}\label{eq:linear}
\sigma\(\left\{e_1,\dots,e_i\right\}\)\subseteq\sigma\(\left\{e_1,\dots,e_j\right\}\)
\end{equation}
for every pair $i<j$. Then the sequence $e_1\dots,e_n$ is linearly ordered.
\end{lemma}

In the next parts of this section we give some examples of greedoids. The exhaustive review of examples of greedoids can be found in~\cite{GKL:1989}. In our article we give only some simplified examples, useful for our aim.

\subsubsection{Trees}
\label{ss:trees}

Let $T$ be a tree with the root $r$ and the set of vertices $V$. 
The distance from the root $r$ to other $v$ is a height $h\(v\)$ of $v$ then $h\(r\)=0$.
The height $h=h\(T\)$ of the tree $T$ is the maximum height of the leaf.

Let $\mathcal{F}$ be the family of all vertex sets such that 
$U\in\mathcal{F}$ if  $U$ is a subtree of $T$ and $r\in U$.
Let 
\begin{equation}\label{eq:tree1closure}
\sigma\(A\)=\bigcap_X\left\{A\subseteq X\subseteq E: E\setminus X\in\mathcal{F} \right\}
\end{equation}
Then $\mathcal{T}=\(V,\mathcal{F}\)$ is a greedoid of feasible sets and $\sigma$ defined by~\eqref{eq:tree1closure} is the closure operator, which fulfils the property $(aex)$. 
Therefore $\mathcal{T}$ is an antimatroid. 
Such an antimatroid can be considered as an example of a hierarchical organisation. Note that the hierarchical structure of dependence in Example~\ref{xmpl:simplest} is the trivial example of a tree (with the element $n$ as a root), and it is a very simple example of antimatroid.

\begin{figure}[!htb]
\centering
\includegraphics[width=6cm]{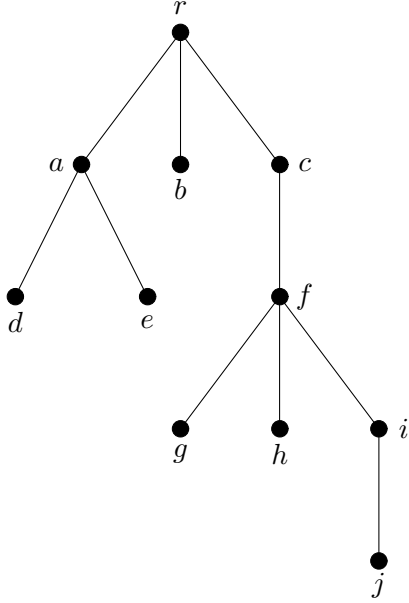}
\caption{\label{fig:tree1}Rooted tree -- greedoid of hierarchical organisation}
\end{figure}
From Theorem~\ref{thm:KV1} and equation \eqref{eq:tree1closure} we have the following result.
\begin{lemma}\label{lem:sigma-tree}
Every closed set $A$ in the given greedoid $\mathcal{T}$ is a sum of $k$ disjoint maximal subtrees $T_i\subseteq T$, $i=1,\dots,k$, with the set of their roots $H=\left\{e_1,\dots,e_k\right\}$ where $e_i\in T_i$ has the highest height in~$T$. 
\end{lemma}
The set $H=H\(A\)$ is the unique spanning set of the set $A$, i.e. is the unique $H$ such that $\sigma\(H\)=A=\sigma\(A\)$. 

In Fig.~\ref{fig:tree1}, for example the sets of vertices $\left\{r,a,d,e\right\}$ and $\left\{r,c,f,g,h\right\}$ belong to $\mathcal{F}$ (they are subtrees  rooted in $r$) but the sets $\left\{d,e\right\}$, $\left\{a,d,e\right\}$ and $\left\{c,g,h,j\right\}$ do not belong to $\mathcal{F}$ (they are not subtrees or they are subtrees do not rooted in $r$).
The set 
\[
A=\left\{a,d,e,f,g,h,i,j\right\}
\]
is closed and with the minimal spanning set $H=\left\{a,f\right\}$.

\subsubsection{Acyclic digraphs}
\label{ss:adigraphs}

Let $D$ be a rooted directed acyclic digraph with the root $r$ and the set of arcs~$E$. 
A~rooted subgraph of $D$ is connected (directionally connected) if for its every vertex $v$ there exist a path from $r$ to $v$.
Let $\mathcal{F}$ be the family of all sets of arcs of connected subgraphs rooted at~$r$. 
\begin{figure}[!htb]
\centering
\includegraphics[width=3.5cm]{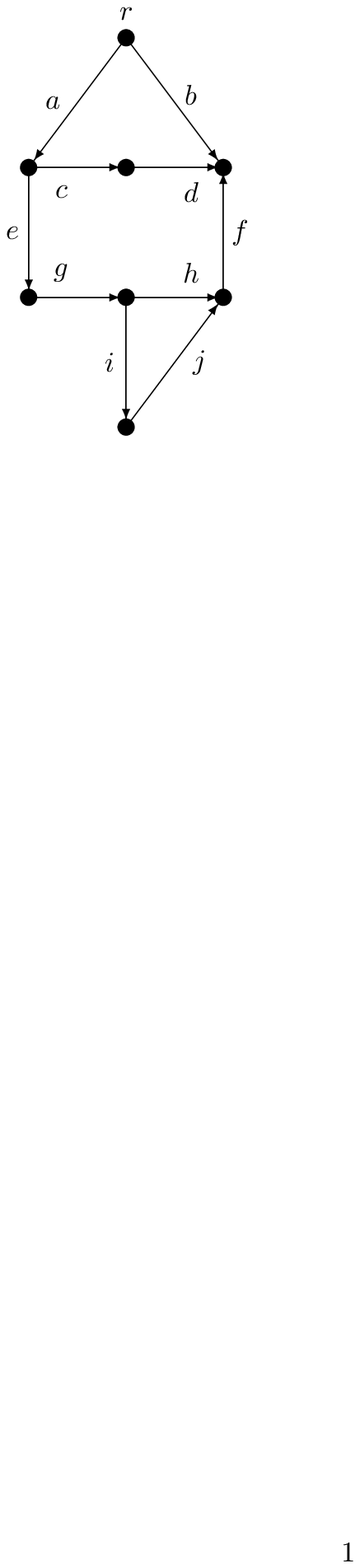}
\caption{\label{fig:digraph1}Rooted acyclic digraph -- greedoid of hierarchical organisation with multiply dependence}
\end{figure}
Let 
\begin{equation}\label{eq:digraph1closure}
\sigma\(A\)=\bigcap_X\left\{A\subseteq X\subseteq E: E\setminus X\in\mathcal{F} \right\}
\end{equation}
Then $\mathcal{D}=\(V,\mathcal{F}\)$ is a greedoid of feasible sets and $\sigma$ defined by~\eqref{eq:digraph1closure} is the closure operator 
(see \cite{Korte_et_al:Greedoids}, p.~26).
Such a greedoid can be considered as an example of a hierarchical organisation with multiple dependencies. 

Note that if every vertex $v\neq r$ has indegree  $d^-\(v\)=1$, the the linegraph of $D$ is a tree. Therefore in such a case, a greedoid $\mathcal{D}$ is isomorphic to a greedoid $\mathcal{T}$ presented in~\ref{ss:trees}.

In Fig.~\ref{fig:digraph1} for example the set of arcs $\left\{a,b,c\right\}$, $\left\{a,b,c,d\right\}$ and $\left\{a,c,e\right\}$ belong to $\mathcal{F}$ (they are connected and rooted in $r$) but $\left\{b,d,f\right\}$ does not belong to $\mathcal{F}$ (it is not rooted in $r$) and $\left\{b,e,j\right\}$ does not belong to $\mathcal{F}$ either (it is not connected, so it is not rooted in $r$).

\subsection{Secretary problem in greedoids}
\label{ss:SPG}

Now we formulate the problem in the most general way, for any greedoid.
Let $\(E,\mathcal{F}\)$, $|E|=n$, be a greedoid with closure operator $\sigma$. On the set $E$ a weight function $w:E\to\mathbb{N}$ is defined. We want to choose the element with the greatest weight under the following conditions.
\begin{enumerate}
\item 
The structure $\(E,\mathcal{F}\)$ and the function $w:E\to\mathbb{R}$, $w\(e\)>0$ for all $e\in E$, is defined but it is not known.
\item 
The elements of $E$ arrive sequentially at the moments $t=1,2,\dots,n$.
\item 
At the moment $t$ we know which element arrives (say the element $e_t$) and we can observe its weight $w\(e_t\)$ and the closure  of $A_{t-1}$ restricted to $A_t$, i.e. $\sigma\(A_{t-1}\)\cap A_t\,$.
\item 
For any two subsets $A',A''\subseteq A$ the possible inclusion $\sigma\(A'\)\subseteq\sigma\(A''\)$ are known.
\item 
Let $A_{t-1}$ be the set of elements which arrived before the moment $t$. If $e_t\in\sigma\(A_{t-1}\)$ then $e_t$ is rejected irrevocably.
\item 
If $e_t\notin\sigma\(A_{t-1}\)$ then we can accept $e_t$ if $w\(e_t\)\geq w\({A_{t-1}}\)$ \\
or reject it in the opposite case. The rejection is irrevocable.
\item 
The process is stopped if the element is accepted or if there  are no next elements to observe.
\end{enumerate}

The proposed algorithm is similar to the algorithm known as Secretary Problem.
\begin{algorithm}\label{alg:algo-greedy}
At each step the observer knows the weight of the chosen element $e$ and performs the actions below:
\begin{enumerate}
\item\label{alg:test-set} 
Fix the closed family of test sets $\mathbf{T}$ or least $\sigma\(\mathbf{T}\)$.
\item
Reject all elements $e_t$ for subsequent $t$ while $A_t\subset \sigma\(\mathbf{T}\)$ ($A_t\neq \sigma\(\mathbf{T}\)$) for some $T$.
\item
For the next $t$ reject it if $w\(e_t\)<w\(A_{t-1}\)$ or $e_t\in\sigma\(A_{t-1}\)$.
\item
If $w\(e_t\)>w\(A_{t-1}\)$ and $e_t\notin\sigma\(A_{t-1}\)$ accept $e_t$ and stop the process.
\end{enumerate}
\end{algorithm}

We take as the criterion the subspaces of the appropriately chosen rank, say rank $k_0$. 
Therefore we unconditionally reject the elements $e_t$ until $\rho\(A_t\)=k_0$. 
For the next $t$ we reject the element $e_t$ if $w\(e_t\)<w\(A_{t-1}\)$ or $e_t\in\sigma\(A_{t-1}\)$.
If $w\(e_t\)>w\({A_{t-1}}\)$ and $e_t\notin\sigma\(A_{t-1}\)$ we accept $e_t$ and stop the process.
To solve this problem we need to determine the distribution of the random variable $\rho\(A_t\)$.  

The presented model in the matroid case is different from the known so far \textit{Matroid Secretary Problem} introduced in \cite{BIKK08}. 
Their model is a generalisation of the multiple choice secretary problem by an additional condition that the chosen set has to be independent. 
In such a model the accepted elements do not have to be independent of the previously rejected elements.
The paper \cite{Soto2013} gives an exhaustive review of known results and presents some new ones.

\section{Special cases}
\label{s:special}

\subsection{Uniform matroid}
\label{ss:uniform}

Uniform matroid $U_{k,n}=\(E,\mathcal{I}\)$, where independent sets are all subsets of $E=\left[n\right]$ with the number of element no greater than $k$:
\begin{align*}
\mathcal{I}=\left\{I:|I|\leq k\right\}, \\
\sigma\(A\)=
\begin{cases}
A & \text{if $|A|<k$,} \\
E & \text{if $|A|\geq k$.}
\end{cases}
\end{align*}
Obviously there must be $k>v$.
Assume  $w\(i\)=i$.
Then the best choice is with probability:
\[
\frac{1}{n}\sum_{i=v}^k\frac{v}{i}
\approx\int\limits_v^k\frac{dx}{x}
=\frac{v}{n}\ln\frac{k}{n},
\]
as in Theorem~\ref{thm:Dynkin}.
The maximum is achieved for $v\approx k/e$ and the optimal probability is $\(k/n\)/e$. If $k=n$ we obtain the classical case with the solution given by Theorem~\ref{thm:Dynkin}.

\subsection{Binary trees}
\label{ss:bitrees}

\begin{definition}
A binary tree with $n$ vertices is an empty tree $T=\emptyset$ if $n=0$ or a triple $T=\(L,r,R\)$ where $r$ is the root of the tree, $L$ (left subtree) is a binary tree with $l$ vertices and $R$ (right subtree) is a binary tree with $p$ vertices, where $n=l+p+1$.
For nonempty $T$, the root of $L$ is called a left child of $r$ and the root of $R$ is called a right child of $r$. If $T=\(\emptyset,v,\emptyset\)$, then $v$ is a leaf.
\end{definition}

\begin{definition}
A complete binary tree is a binary tree in which all nodes other than the leaves have two children. If moreover all leaves have the same height, the binary tree is complete and full. 
\end{definition}
The number $l$ of leaves in a complete and full binary tree with $n$ vertices is $l=\(n+1\)/2=2^h$. Thus $n=2^{h+1}-1$ is the number of vertices of such a tree. 
The sequence $v_1,\dots,v_k$ is linear if $h\(v_{i+1}\)=h\(v_i\)+1$.

Similarly to Example~\ref{xmpl:simplest} we will consider two different cases.
First, let us consider the case $w\(v\)=h-h\(v\)+1$. Therefore the root $r$ has the maximal weight $w\(r\)=h+1$ and leaves $u$ have the minimal weights $w\(u\)=1$. 

In the second case we assume that
\begin{enumerate}
\item 
the set of weights has exactly $h+1$ values,
\item 
exactly $j+1$ vertices have the value $w\(j\)$,
\item 
$w\(0\)>w\(1\)>\dots>w\(h\)$,
\item 
values are equally likely distributed on all $n=2^{h+1}-1$ vertices.
\end{enumerate}
 
Similar to our Case~1 there is the known model which was considered by \cite{Morayne:Partial-DM1998}. 
Instead of closure $\sigma\(A\)$ used in Algorithm~\ref{alg:algo-greedy} the procedure used  in this model checks whether  
\begin{itemize}
\item
$w\(e_k\)>\max\left\{w\(v_1\),\dots,w\(v_{k-1}\)\right\}$ and 
\item
$\left\{v_1,\dots,v_{k-1}\right\}$ not linear or it is linear and $k>h/2$. 
\end{itemize}
The element $e_k$ is accepted if both of the above conditions are fulfilled.
\begin{theorem}[\cite{Morayne:Partial-DM1998}]
\label{thm:Morayne1998}
Algorithm~\ref{alg:algo-greedy} gives an optimal strategy for the choice of the element $r$, i.e. $\Pr\(v_k=r\)$ is the maximal possible. 
If $h\to\infty$ then the Algorithm~\ref{alg:algo-greedy} gives an optimal choice with probability tending to 1.
\end{theorem}

\subsection{Graphical matroids}
\label{ss:graph-matr}

\subsubsection{Graphical model of secretary problem}
\label{sss:graph_model}

Let $G=\(V,E\)$ be an undirected graph where $V$ is the set of vertices and $E$ is the set of edges. 
An independent set is any set of edges which does not contain any cycles, i.e. the independent set forms a forest.
In this section only the case $G=K_n$, where $K_n$ is an $n$-vertices complete graph, is considered.

The random graph introduced by \cite{ER:evolution} is constructed by connecting nodes randomly. Since that time many monographs and textbooks have been devoted to the theory of random graphs. Among others we refer the reader to the following books: \cite{Bollobas:random-graphs},  \cite{JLR:random-graphs} and \cite{Hofstad:random-graphs}.

In this paper we will consider the so called ``random graph process'' (see~\cite{JLR:random-graphs}, p.~4). 
Let $n$, a number of vertices be fixed. Let $G_{n,t}$ be any fixed graph with $n$ vertices and $t$ edges.
The random graph process is a stochastic process which begins with no edges at time $t=0$ and adds new edges, one at time; each new edge is selected at random, uniformly among all edges not presented until now. At the moment $t$, $0\leq t\leq m=\binom{n}{2}$, the random graph $G_n\(t\)$ has $t$ edges and 
\[
\Pr\(G_n\(t\)=G_{n,t}\)=\binom{m}{t}^{-1}.
\] 
Let us consider the asymptotic case where $n\to\infty$ and $t=t\(n\)$. 
To simplify the notation, use the abbreviation \textit{a.a.s} (asymptotic almost surely) instead of the term ``with the probability tending to 1 when $n\to\infty$''.
If $k>2$ and $t\ll n$ but $t=n^{1-o\(1\)}$ then \textit{a.a.s}, $G_n\(t\)$ has no cycles, i.e. the set of edges forms an independent set (see~\cite{JLR:random-graphs}, p.~104). 
This means that the beginning of such the process is similar to the beginning of the process without dependence restrictions. 
Nevertheless the number $t=n^{1-o\(1\)}$ of tested elements is too small to obtain a reasonable decision.

In order to change to the proper range of numbers of edges which give a sufficient information to obtain an optimal decision, we have to consider such a case, where the number $t$ of tested edges is big enough and furthermore the number of edges which are not dependent tested as well as the number of rejected edges are also  big enough.
Such a situation is given by the following fundamental result, proved by Erd{\H{o}}s and R{\'e}nyi in their famous paper \cite{ER:evolution} (see \cite{JLR:random-graphs} and \cite{Hofstad:random-graphs}).
\begin{theorem}\label{thm:eq:phase_trans}
If 
\begin{equation}\label{eq:phase_trans}
t\sim \frac{n\(\ln n+\ln\lambda\)}{2},
\end{equation}
where $\lambda>0$, then the random graph \textit{a.a.s.} has one giant component and $N$ isolated vertices. 
The random variable $N$ has Poisson distribution with the mean $\lambda$.
\end{theorem}
For the big $\lambda$ we can obtain a better balance between a number of tested elements (given by Eq.~\eqref{eq:phase_trans}) and a number of edges possible to choose, i.e. edges which do not belong to the giant component. 
From Theorem~\ref{thm:eq:phase_trans}, the giant component has \textit{a.a.s.} $n\(\ln n+\ln\lambda\)$ elements (edges) and the rank $n-N$. 
Because every new edge \textit{a.a.s.} joins an isolated vertex with the giant component then we can choose an optimal $k$ elements set from $N\(n-N\)$ elements.

From $3\sigma$ rule we have approximately $\Pr\(N-\lambda>3\sqrt{\lambda}\)\leq 0.005$. 
To obtain $\rho\(T\)=n-\lambda$ we should \textit{a.a.s.} test at least 
\[
t_0\sim\frac{n\(\ln n+\ln\lambda\)}{2}-3\sqrt{\lambda}
\]
elements plus perhaps an additional next $3\sqrt{\lambda}/2$ elements.
\begin{example}\label{xmpl:Poisson}
In Table~\ref{tab:Poisson} there are shown the values of the necessary number $t_0$ of testing steps to achieve the set $T$ of rank given before. Let $\rho\(T\)=\lambda$.

\begin{table}[!hbt]
\centering 
\caption{\label{tab:Poisson}Necessary number $t_0$ of testing steps for given $n$ and $\lambda$.}

\smallskip
\begin{tabular}{r|r|r|r|r}
$n$ & $\lambda=100$ & $\lambda=200$ & $\lambda=300$ & $\lambda=400$ \\\hline 
 1000 & 11482 & 12163 & 12559 & 12839 \\
 2000 & 24382 & 25756 & 26557 & 27124 \\
 3000 & 37804 & 39871 & 41078 & 41933 \\
 5000 & 65581 & 69035 & 71052 & 72483 \\
10000 & 138125 & 145044 & 149089 & 151958
\end{tabular}
\end{table}
\end{example}
Note, that after rejecting approximately next $\lambda$ edges after the moment $t_0$, $\rho\(R\)\sim n$, the process will be finished. 
If all values $w\(e_j\)$ are different for $j=1,2,\dots,\binom{n}{2}$, then it is clear that the probability of choosing the  optimal solution (the edge of maximal weight or the set of $k$ edges with maximal sum of weights) rapidly tends to zero.

\subsubsection{Linearly decreasing number of linearly ordered weights}
\label{sss:limited}

Let us assume that there exist only $n-1$ values of weights of edges in the $n$-vertices graph. 
In this section we restrict ourselves to the case $k=1$, i.e. to the choice of only one, the best element.
Without the loss of generality one can assume that $w\(e\)\in\left\{1,2,\dots,n-1\right\}$ for all edges of the graph $K_n$.
Similarly to Example~\ref{xmpl:simplest} we consider the three completely different cases.
\begin{enumerate}
\item\label{enum:graph_lin_inc} 
Let $V=\left\{1,2,\dots,n\right\}$ be the set of vertices and $e_{ij}=\left\{i,j\right\}$.
For $i=1,2,\dots,n-1$ and $k=i+1,\dots,n$ let $w\(e_{i,k}\)=k$.
\item\label{enum:graph_lin_dec} 
Let $V=\left\{1,2,\dots,n\right\}$ be the set of vertices and $e_{ij}=\left\{i,j\right\}$.
For $i=1,2,\dots,n-1$ and $k=i+1,\dots,n$ let $w\(e_{i,k}\)=n-k$.
\item\label{enum:graph_rnd}  
Every value appears approximately $n/2$ times and these values are distributed equally likely.
\end{enumerate}

At first, let us consider Case~\ref{enum:graph_lin_inc}. 
In this case we have only one best element, but $n-1$ the worst element.
If the maximal element belongs to the giant component, then the optimal solution does not exist.
In Case~\ref{enum:graph_lin_dec} we have $n-1$ the best elements, but only one worst element. 
If the maximal element belongs to the giant component, then the optimal solution does not exit.

\section{Prospective application: cloud computing}
\label{s:cloud}

It is obvious that the simplest model closely related with the name \textit{Secretary Problem} is very far from real applications. In this section we describe the simplified, but more realistic  model of cloud computing, which can used as an example of an application to the computer networks%
\footnote{This applications was inspired by problems arisen during the realisation of the grant \textit{Research on cloud based distribution and management technology of software and licenses for research and science units}}. Below we shortly describe the model.

Cloud computing is definitely one of the fastest developing technologies in IT sector. 
Year by year this kind of solutions become more popular. 
This idea has actually its implementations in many different models. 
Regardless of the fact which of them is used the general idea is still the same: most of the duties related to IT infrastructure maintenance is moved from the user (customer) to the service provider. 
In other words we can say that the same classical element (i.e. server or software running on it) becomes just a service, available for the user by the computer network. 
The user, who has a task to be  performed, just orders the resources needed for this particular time. 
This solution is very comfortable for the user as more efficient resources usage guarantees also economic benefits. 
Since in typical cloud computing service many different users share with each other limited hardware and software resources, optimisation of their utilisation is the key problem.

Let us consider the situation, where the user has same the computing task to be performed in the shortest possible time. 
To do this job, a virtual machine with required hardware resources (computing cores, RAM memory etc.) must be rented. 
Then there is a need to deliver a significant amount of data required for computing. This operation is strictly related to the transfer time. 
Some parameters of the virtual machine are simple to compare (results of popular benchmarks, user estimation based on declared hardware parameters). 
In the real environment also some other parameters, often difficult for forecasting, should also be considered. 
One of them is an actually available throughput of the computer network between the client host and the computing node. 
While the bandwidth can be considered constant, the throughput is directly connected with the current utilisation of the network. 
Due to the above, time of transfer can be approximated no sooner than after sending a few TCP datagrams and receiving acknowledgements. 
At the moment when the transmission speed would classified as unsatisfactory, it can be interrupted and the next localisation can be considered. 
However, what is very important at the time of the resignation of the given service provider, the resources can be assigned to other tasks, and they are not available anymore. 
What is more, there can be some relationship between individual service providers. Their hardware resources can be located in the same network segments. Therefore, the rejection of one or more of the service providers in the network should also result in the elimination of other nodes located in the same network location and depended of rejected nodes.

Assuming that the systems work in a such the way that at each step they try to choose the best node, then our model (matroid and more generally -- greedoid) can be applied as a model of activities in the cloud. 
Certainly, the accurate choice needs deeper considerations and verifications with the real networks and their management.

\section{Conclusion}
\label{s:conlu}
We presented a model of optimal choice among objects which are connected by different dependencies. 
Our aim is to choose an object or a set object but in a such way so that the chosen objects were independent in some sense. 
The independence in the model is described in the term of greedoids and as special cases -- matroids, antimatroids and more special cases, for example rooted trees and random graphs.
As the first step we try to apply such models to a more realistic problem, namely to the problem of operations during the cloud computing.

\bigskip 
{
\small 
\noindent 
\textbf{Acknowledgements.}
The author thanks Piotr Nadybski who helped to formulate the problem of operation in a cloud and wrote the most part of Section~\ref{s:cloud}.
}

\addcontentsline{toc}{section}{References}
\bibliographystyle{abbrvnat}
\bibliography{mrabbrev,combintr,matroids,graphs,reliab,probab,rndgraph,rndstruc,varia}

\begin{thebibliography}{20}
\providecommand{\natexlab}[1]{#1}
\providecommand{\url}[1]{\texttt{#1}}
\expandafter\ifx\csname urlstyle\endcsname\relax
  \providecommand{\doi}[1]{doi: #1}\else
  \providecommand{\doi}{doi: \begingroup \urlstyle{rm}\Url}\fi

\bibitem[Babaioff et~al.(2008)Babaioff, Immorlica, Kempe, and
  Kleinberg]{BIKK08}
M.~Babaioff, N.~Immorlica, D.~Kempe, and R.~Kleinberg.
\newblock Online auctions and generalized secretary problems.
\newblock \emph{SIGecom Exch.}, 7\penalty0 (2):\penalty0 7:1--7:11, June 2008.
\newblock ISSN 1551-9031.
\newblock \doi{10.1145/1399589.1399596}.
\newblock URL \url{http://doi.acm.org/10.1145/1399589.1399596}.

\bibitem[Bollob\'as(2001)]{Bollobas:random-graphs}
B.~Bollob\'as.
\newblock \emph{Random Graphs}.
\newblock Academic Press, London, 2001.

\bibitem[Dynkin(1963)]{Dynkin:opt}
E.~B. Dynkin.
\newblock The optimum choice of the instant for stopping a {M}arkov process.
\newblock \emph{Soviet Math. Dokl.}, 4:\penalty0 627--629, 1963.

\bibitem[Erd\H{o}s and R{\'e}nyi(1960)]{ER:evolution}
P.~Erd\H{o}s and A.~R{\'e}nyi.
\newblock On the evolution of random graphs.
\newblock \emph{Publications of the Mathematical Institute of the Hungarian
  Academy of Sciences}, 5:\penalty0 17--61, 1960.

\bibitem[Ferguson(1989)]{Ferguson:Who}
T.~S. Ferguson.
\newblock Who solved the secretary problem?
\newblock \emph{Statistical Science}, 4\penalty0 (3):\penalty0 282--289, 1989.

\bibitem[Girdhar and Dudek(2009)]{GirdharDudek:How}
Y.~Girdhar and G.~Dudek.
\newblock Optimal online data sampling or how to hire the best secretaries.
\newblock In \emph{Canadian Conference on Computer and Robot Vision}, pages
  292--298, 2009.

\bibitem[Goecke et~al.(1989)Goecke, Korte, and Lov{\'a}s]{GKL:1989}
O.~Goecke, B.~Korte, and L.~Lov{\'a}s.
\newblock Examples and algorithmic properties of greedoids.
\newblock In B.~Simeone, editor, \emph{Combinatorial optimization}, volume 1403
  of \emph{Lecture Notes in Mathematics}, pages 113--161. Springer, 1989.

\bibitem[Hajiaghayi et~al.(2004)Hajiaghayi, Kleinberg, and
  Parkes]{Hajiaghayi_et_al:Adptive}
M.~T. Hajiaghayi, R.~Kleinberg, and D.~C. Parkes.
\newblock Adaptive limited-supply online auctions.
\newblock In \emph{EC'04:Proceedings of the 5th ACM Conference on Electronic
  Commerce}, pages 71--80, New York, 2004. ACM Press.

\bibitem[Janson et~al.(2000)Janson, {\L}uczak, and
  Ruci\'nski]{JLR:random-graphs}
S.~Janson, T.~{\L}uczak, and A.~Ruci\'nski.
\newblock \emph{Random Graphs}.
\newblock Wiley, New York, 2000.

\bibitem[Kleinberg(2005)]{Kle05}
R.~Kleinberg.
\newblock A multiple-choice secretary algorithm with applications to online
  auctions.
\newblock In \emph{Proceedings of the sixteenth annual ACM-SIAM symposium on
  Discrete algorithms}, SODA '05, pages 630--631, Philadelphia, PA, USA, 2005.
  Society for Industrial and Applied Mathematics.
\newblock ISBN 0-89871-585-7.
\newblock URL \url{http://dl.acm.org/citation.cfm?id=1070432.1070519}.

\bibitem[Klimesch(1994)]{Klimesch:Structure}
W.~Klimesch.
\newblock \emph{The Structure of Long-term Memory: A Connectivity Model of
  Semantic Processing}.
\newblock Lawrence Erlbaum Associates, Inc., Publishers, 1994.

\bibitem[Korte and Lov\'as(1983)]{KorteLovas:structural}
B.~Korte and L.~Lov\'as.
\newblock Structural properties of greedoids.
\newblock \emph{Combinatorica}, 3--4\penalty0 (3):\penalty0 359--374, 1983.

\bibitem[Korte and Vygen(2012)]{KorteVygen:CombOpti5}
B.~Korte and J.~Vygen.
\newblock \emph{Combinatorial Optimization}, volume~21 of \emph{Algorithms and
  Combinatorics}.
\newblock Springer-Verlag, Berlin Heidelberg, 5 edition, 2012.

\bibitem[Korte et~al.(1991)Korte, Lov{\'a}sz, and
  Schrader]{Korte_et_al:Greedoids}
B.~Korte, L.~Lov{\'a}sz, and R.~Schrader.
\newblock \emph{Greedoids}, volume~4 of \emph{Algorithms and Combinatorics}.
\newblock Springer-Verlag, Berlin Heidelberg, 1991.

\bibitem[Morayne(1998)]{Morayne:Partial-DM1998}
M.~Morayne.
\newblock Partial order analogue of the secretary problem: the binary tree
  case.
\newblock \emph{Discrete Math.}, 184:\penalty0 165–181, 1998.

\bibitem[Oxley(2011)]{Oxley:ma2nd}
J.~G. Oxley.
\newblock \emph{Matroid Theory}.
\newblock Oxford University Press, Oxford, 2 edition, 2011.

\bibitem[Soto(2013)]{Soto2013}
J.~A. Soto.
\newblock Matroid secretary problem in the random-assignment model.
\newblock \emph{SIAM J. Comput.}, 42\penalty0 (1):\penalty0 178--211, 2013.

\bibitem[van~der Hofstad(2016)]{Hofstad:random-graphs}
R.~van~der Hofstad.
\newblock \emph{Random Graphs and Complex Networks}, volume~I.
\newblock 2016.
\newblock URL \url{https://www.win.tue.nl/~rhofstad/NotesRGCN.pdf}.

\bibitem[Welsh(1976)]{Welsh:ma}
D.~J.~A. Welsh.
\newblock \emph{Matroid Theory}.
\newblock Academic Press, London, 1976.

\bibitem[Wilson(2010)]{Wilson-eng5}
R.~J. Wilson.
\newblock \emph{Introduction to Graph Theory}.
\newblock Prentice Hall, 5 edition, 2010.

\end{thebibliography}

\end{document}